\long\def\comment#1{}
\documentclass[onecolumn]{IEEEtran}
\usepackage{xspace,amsmath,epsfig,syntonly,psfrag}
\usepackage{wrapfig}

\DeclareSymbolFont{AMSa}{U}{msa}{m}{n}
\DeclareMathSymbol{\blacksquare}  {\mathord}{AMSa}{"04}

\newfont{\bb}{msbm10 scaled 1100}

\newcommand{\CC}{{\mathchoice {\setbox0=\hbox{$\displaystyle\rm C$}\hbox{\hbox
to0pt{\kern0.4\wd0\vrule height0.9\ht0\hss}\box0}}
{\setbox0=\hbox{$\textstyle\rm C$}\hbox{\hbox
to0pt{\kern0.4\wd0\vrule height0.9\ht0\hss}\box0}}
{\setbox0=\hbox{$\scriptstyle\rm C$}\hbox{\hbox
to0pt{\kern0.4\wd0\vrule height0.9\ht0\hss}\box0}}
{\setbox0=\hbox{$\scriptscriptstyle\rm C$}\hbox{\hbox
to0pt{\kern0.4\wd0\vrule height0.9\ht0\hss}\box0}}}}

\newcommand{\RR}{{\mathchoice {\setbox0=\hbox{$\displaystyle\rm R$}\hbox{\hbox
to0pt{\kern0.4\wd0\vrule height0.9\ht0\hss}\box0}}
{\setbox0=\hbox{$\textstyle\rm R$}\hbox{\hbox
to0pt{\kern0.4\wd0\vrule height0.9\ht0\hss}\box0}}
{\setbox0=\hbox{$\scriptstyle\rm R$}\hbox{\hbox
to0pt{\kern0.4\wd0\vrule height0.9\ht0\hss}\box0}}
{\setbox0=\hbox{$\scriptscriptstyle\rm R$}\hbox{\hbox
to0pt{\kern0.4\wd0\vrule height0.9\ht0\hss}\box0}}}}

\newcommand{\QQ}{{\mathchoice {\setbox0=\hbox{$\displaystyle\rm Q$}\hbox{\hbox
to0pt{\kern0.4\wd0\vrule height0.9\ht0\hss}\box0}}
{\setbox0=\hbox{$\textstyle\rm Q$}\hbox{\hbox
to0pt{\kern0.4\wd0\vrule height0.9\ht0\hss}\box0}}
{\setbox0=\hbox{$\scriptstyle\rm Q$}\hbox{\hbox
to0pt{\kern0.4\wd0\vrule height0.9\ht0\hss}\box0}}
{\setbox0=\hbox{$\scriptscriptstyle\rm Q$}\hbox{\hbox
to0pt{\kern0.4\wd0\vrule height0.9\ht0\hss}\box0}}}}

\newcommand{\ZZ}{{\mathchoice {\setbox0=\hbox{$\displaystyle\rm Z$}\hbox{\hbox
to0pt{\kern0.4\wd0\vrule height0.9\ht0\hss}\box0}}
{\setbox0=\hbox{$\textstyle\rm Z$}\hbox{\hbox
to0pt{\kern0.4\wd0\vrule height0.9\ht0\hss}\box0}}
{\setbox0=\hbox{$\scriptstyle\rm Z$}\hbox{\hbox
to0pt{\kern0.4\wd0\vrule height0.9\ht0\hss}\box0}}
{\setbox0=\hbox{$\scriptscriptstyle\rm Z$}\hbox{\hbox
to0pt{\kern0.4\wd0\vrule height0.9\ht0\hss}\box0}}}}

\newcommand{\KK}{{\mathchoice {\setbox0=\hbox{$\displaystyle\rm K$}\hbox{\hbox
to0pt{\kern0.4\wd0\vrule height0.9\ht0\hss}\box0}}
{\setbox0=\hbox{$\textstyle\rm K$}\hbox{\hbox
to0pt{\kern0.4\wd0\vrule height0.9\ht0\hss}\box0}}
{\setbox0=\hbox{$\scriptstyle\rm K$}\hbox{\hbox
to0pt{\kern0.4\wd0\vrule height0.9\ht0\hss}\box0}}
{\setbox0=\hbox{$\scriptscriptstyle\rm K$}\hbox{\hbox
to0pt{\kern0.4\wd0\vrule height0.9\ht0\hss}\box0}}}}

\newcommand{\EE}{{\mathchoice {\setbox0=\hbox{$\displaystyle\rm E$}\hbox{\hbox
to0pt{\kern0.4\wd0\vrule height0.9\ht0\hss}\box0}}
{\setbox0=\hbox{$\textstyle\rm E$}\hbox{\hbox
to0pt{\kern0.4\wd0\vrule height0.9\ht0\hss}\box0}}
{\setbox0=\hbox{$\scriptstyle\rm E$}\hbox{\hbox
to0pt{\kern0.4\wd0\vrule height0.9\ht0\hss}\box0}}
{\setbox0=\hbox{$\scriptscriptstyle\rm E$}\hbox{\hbox
to0pt{\kern0.4\wd0\vrule height0.9\ht0\hss}\box0}}}}

\newcommand{\FF}{{\mathchoice {\setbox0=\hbox{$\displaystyle\rm F$}\hbox{\hbox
to0pt{\kern0.4\wd0\vrule height0.9\ht0\hss}\box0}}
{\setbox0=\hbox{$\textstyle\rm F$}\hbox{\hbox
to0pt{\kern0.4\wd0\vrule height0.9\ht0\hss}\box0}}
{\setbox0=\hbox{$\scriptstyle\rm F$}\hbox{\hbox
to0pt{\kern0.4\wd0\vrule height0.9\ht0\hss}\box0}}
{\setbox0=\hbox{$\scriptscriptstyle\rm F$}\hbox{\hbox
to0pt{\kern0.4\wd0\vrule height0.9\ht0\hss}\box0}}}}

\newcommand{\GG}{{\mathchoice {\setbox0=\hbox{$\displaystyle\rm G$}\hbox{\hbox
to0pt{\kern0.4\wd0\vrule height0.9\ht0\hss}\box0}}
{\setbox0=\hbox{$\textstyle\rm G$}\hbox{\hbox
to0pt{\kern0.4\wd0\vrule height0.9\ht0\hss}\box0}}
{\setbox0=\hbox{$\scriptstyle\rm G$}\hbox{\hbox
to0pt{\kern0.4\wd0\vrule height0.9\ht0\hss}\box0}}
{\setbox0=\hbox{$\scriptscriptstyle\rm G$}\hbox{\hbox
to0pt{\kern0.4\wd0\vrule height0.9\ht0\hss}\box0}}}}

\newfont{\BBB}{msbm10 scaled 2200}

\newcommand{\uv}{{\bf u}}
\newcommand{\wv}{{\bf w}}
\newcommand{\vv}{{\bf v}}

\newcommand{\yv}{{\bf y}}

\newcommand{\zerov}{{\bf 0}}

\newcommand{\Gm}{{\bf G}}
\newcommand{\Hm}{{\bf H}}
\newcommand{\Id}{{\bf I}}

\newcommand{\Mm}{{\bf M}}

\newcommand{\Sm}{{\bf S}}
\newcommand{\Tm}{{\bf T}}

\newcommand{\Wm}{{\bf W}}

\newcommand{\Ym}{{\bf Y}}

\newcommand{\Ac}{{\cal A}}

\newcommand{\Lc}{{\cal L}}

\newcommand{\Sc}{{\cal S}}

\newcommand{\Mcb}{\boldmath {\cal M}}

\newcommand{\transp}{{{\sf T}}}

\long\def\comment#1{}
\newlength{\lwidth}%
\newcommand{\lettrine}[2]{%
\settowidth{\lwidth}{#2\kern2pt}%
\noindent\hangindent\lwidth\hangafter-#1\hskip-\lwidth%
\smash{\hbox to\lwidth{\raise7pt\vtop{\null\hbox{#2}}%
\hfill}}\ignorespaces}
\newfont{\HUGEfont}{cmr17 scaled \magstep5}  
\newcommand{\defines}{{\,\,\stackrel{\scriptscriptstyle \bigtriangleup}{=}\,\,}}
\def\appendixname{Appendix}
\newcommand{\beq}{\begin{equation}}
\newcommand{\enq}{\end{equation}}
\newcommand{\beqa}{\begin{eqnarray}}
\newcommand{\enqa}{\end{eqnarray}}

\newcommand{\IC}{\mbox{\bb C}}
\newcommand{\IF}{\mbox{\bb F}}

\newcommand{\IQ}{\mbox{\bb Q}}

\newtheorem{theorem}{Theorem}

\newenvironment{example}{{\bf Example}.}{\vspace{\baselineskip}}

\newcommand{\be}{\begin{eqnarray*}}
\newcommand{\ee}{\end{eqnarray*}}
\newcommand{\ben}{\begin{eqnarray}}
\newcommand{\een}{\end{eqnarray}}

\begin{document}

\title{A New Class of TAST Codes With A Simplified Tree Structure}

\author{
{Mohamed Oussama Damen$^1$,}
{Hesham El Gamal$^2$ }
{and Ahmed A. Badr$^3$}
\thanks{\protect\rule{0pt}{1.5em}$^1$ M.O. Damen is with the ECE department, University of Waterloo, On, Canada (modamen@ece.uwaterloo.ca).}
\thanks{\protect\rule{0pt}{1.5em}$^2$ H. El Gamal is with the
  Department of Electrical Engineering, the Ohio State University,
  Columbus, OH (helgamal@ee.eng.ohio-state.edu).}
\thanks{\protect\rule{0pt}{1.5em}$^3$ A.A. Badr is with the Wireless Intelligent Networks Center (WINC), Nile University, Cairo, Egypt (ahmed.atef.younes@gmail.com).}}

\renewcommand{\thepage}{}
\maketitle
\pagenumbering{arabic}


\begin{abstract}

We consider in this paper the design of full diversity and high rate space-time codes with moderate decoding complexity for arbitrary number of transmit and receive antennas and arbitrary input alphabets. We focus our attention to codes from the threaded algebraic space-time (TAST) framework since the latter includes most known full diversity space-time codes. We propose a new construction of the component single-input single-output (SISO) encoders such that the equivalent code matrix has an upper triangular form. We accomplish this task by designing each SISO encoder to create an ISI-channel in each thread. This, in turn, greatly simplifies the QR-decomposition of the composite channel and code matrix, which is essential for optimal or near-optimal tree search algorithms, such as the sequential decoder. 
\end{abstract}

\section{Introduction and System Model}\label{Intro}
Since the beginning of research on multiple-input multiple-output (MIMO) systems, there has been a growing interest in the construction of low complexity space-time coding (e.g., 
\cite{Tarokh:al98,alamouti,tirkkonen,liang,elgamal:damen02,jafarkhani,silver,Rajan1,Rajan2,hikmet} and references therein). However, high rate and low complexity are so far limited to two transmit antennas \cite{alamouti,silver,hikmet}, whereas for higher number of transmitters, the constructed codes are either low rate (like orthogonal or quasi-orthogonal designs \cite{alamouti,tirkkonen,jafarkhani,Rajan2}), have long delays (e.g., \cite{liang}) or have a considerable higher complexity than the orthogonal designs \cite{elgamal:damen02,Rajan1}. 
In this work, we investigate the Threaded Algebraic Space-Time (TAST) framework, proposed by the authors in \cite{elgamal:damen02} and which is shown to comprise most full diversity codes (e.g., orthogonal designs, codes from cyclic division algebra, etc. \cite{simpletast06}). The TAST codes are full rate, full diversity and decodable by algorithms that search trees for shortest paths such as sequential decoding. The tree of the coded MIMO system is usually constructed via QR decomposition of the composite channel and code matrix \cite{dec03}. 
Within the TAST framework, we investigate here the construction of codes that facilitate tree-searching decoding algorithms by simplifying the tree structure of the new proposed TAST codes.  

For our purposes, we consider a space-time block code ${\Sc}$ over $M$ transmit
antennas and $T$ channel uses.  The code words of $\Sm=\left[s_{m,t}\right]$ are $M\times T$ matrices in
which the $(m,t)$-th entry $s_{m,t}$ denotes the modulated code
symbol transmitted from the $m$-th transmit antenna at
the $t$-th symbol periods. If all nontrivial pairwise differences
between code word matrices are full rank (i.e., of rank $M$) over the complex
numbers $\IC$, then the space-time code is said to achieve full diversity.  

In this work, the transmitted symbols are 
generated from some finite constellations using algebraic
constructions (e.g., \cite{elgamal:damen02}). 
Let $\Ac$ denote the two-dimensional
quadrature amplitude modulation (QAM) alphabet, and let $\IF =
\IQ(i)$, ($i=\sqrt{-1})$, denote the field of complex rational
numbers. We let 
choose $\IF(\theta)$ be an extension field of degree
$[\IF(\theta):\IF]$ over the base field $\FF$ that contains $\Ac$. The rate of the code $\Sc$ is defined as the number of complex symbols from $\Ac$ transmitted per channel use.  

The paper is organized as follows. In Section \ref{TAST} we review the construction methodology of the TAST framework and we present our new construction in Section \ref{ST-TAST} before concluding in Section \ref{conc}. 

\section{The TAST Framework}\label{TAST}

\subsection{Codes Construction}
The basic idea of the TAST framework is to divide the space-time matrix into several threads (i.e., permutation matrices or concatenation of permutation matrices) and then assign a different full diversity SISO encoder to each thread. That makes each thread achieves full diversity in the absence of the other threads. Then, to achieve full diversity of the total code, one needs to project each SISO encoder into a different algebraic subspace by multiplying its elements by some roots of unity 
\cite{elgamal:damen02}. 
The threaded code structure was introduced by El Gamal and Hammons
\cite{bib:ElGHam01} as a generalization of the layers proposed by
Foschini \cite{Foschini96}.  In the aforementioned framework, a {\em layer\/} is a function
that assigns a unique transmit antenna to be used during the
symbol intervals that are available over a code word interval. A
layer is called a {\em thread\/} if the domain of the function
includes all possible symbol intervals and its range includes all
possible transmit antennas (i.e., a thread exploits all the spatial diversity of a MIMO system). 
For example, a MIMO system with $M$
transmit antennas and a code word interval of duration $T$ symbols
allows the simultaneous transmission of $M$ different threads.

An example of the threaded layering set 
${\Lc}$ is given in the following equation (with the convention that time indices span $[0,\,
T-1]$):
\begin{eqnarray}\label{thread-layering}
\ell_j = \left\{ (\left\lfloor t+j-1\right\rfloor_M + 1, t) : 0\le
t<T\right\} \mbox { for } 1\leq j \leq L, \label{thread}
\end{eqnarray}
where 
$\lfloor \cdot \rfloor_M$ denotes the mod-$M$ operation. 

Then, a TAST code consists of sending different SISO encoders into different threads.
For example, for $T=M$, full diversity full rate TAST codes can be constructed using the following SISO encoders \cite{elgamal:damen02}:  
\[
\gamma_j = \phi_j \Mm_j \uv_j, \, \, \, j=1,\ldots, M, 
\]
where $\Mm_j$ is an $M\times M$ full diversity constellation rotation, $\uv_j$ is an $M\times 1$ an information symbol vector with components from a given QAM constellation and $\phi_1, \ldots, \phi_j$ are some roots of unity that are algebraically independent over the number field containing the rotations $\Mm_j, j=1, \ldots, M$. 

The TAST framework contains most of the known full diversity space-time codes, such as orthogonal designs and codes from cyclic division algebra for example \cite{simpletast06}.  
An example of TAST codes is the Golden code (or its variant Matrix C in the WiMax standards)
\begin{eqnarray}\label{goldencode}
 \Gm \defines
\frac{1}{\sqrt{5}}\left[\begin{array}{cc}
\alpha(u_1 + \theta u_4) & \phi^{1/2} \sigma(\alpha(u_2+\theta u_3)) \\
\phi^{1/2} \alpha(u_2 + \theta u_3) & \sigma(\alpha(u_1 + \theta u_4))
\end{array}\right], \\ 
\nonumber \phantom{\defines} u_1, u_2, u_3, u_4 \in \, \mbox{QAM}
\end{eqnarray}
where  $\theta=\frac{1+\sqrt{5}}{2}$, $\phi =i$,
$\sigma(\theta) = 1-\theta$, and $\alpha = 1+i (1-\theta)$ is a
shaping coefficient in order to make the mapping unitary \cite{Golden,mahesh}.

\subsection{Decoding of TAST Codes}
Consider a TAST code with rate $R= \frac{K}{T}$ complex symbols per channel use, i.e., each $K\times 1$ QAM information symbol vector $\uv$ is associated with an $M\times T$ code word matrix $\Sm(\uv)$. When transmitting the TAST code word $\Sm(\uv)$ over a MIMO channel, the baseband received signal can be written in the form
\begin{equation}
\Ym = \Hm \Sm(\uv) + \Wm
\end{equation}
where $\Hm$ is the $N\times M$ channel matrix with i.i.d Rayleigh fading coefficients and $\Wm$ is an $N\times T$ matrix representing the additive white Gaussian noise with i.i.d. components. Writing the received signal in a vector form gives
\begin{equation}
\yv \defines \mbox{vec}({\Ym}) = (\Id_T \otimes \Hm) \, \Mcb \uv + \wv
\end{equation}
where $\Id_T$ is the $T\times T$ identity matrix, $\otimes$ is the Kronicker product, $\Mcb$ is the equivalent $MT \times K$ code matrix such that $\Mcb \uv = \mbox{vec}(\Sm(\uv))$. Maximum likelihood (ML) or near-ML decoding can then be done using QR decomposition on the composite $NT \times K$ matrix  $(\Id_T \otimes \Hm) \, \Mcb$ followed by sphere decoding or other type of sequential decoding algorithms that search for the shortest path on a tree such as the M-algorithm \cite{dec03}. 
Since the complexity of the QR decomposition is cubic in the matrix dimension, it can be prohibitive for large values of $T$ and $K$ (e.g., for full rate minimum delay TAST codes, $T=M$ and $K=M^2$). 
In the next section, we propose a new construction of TAST codes that simplifies the QR decomposition of $(\Id_T \otimes \Hm) \, \Mcb$ by designing $\Mcb$ to be in an upper triangular form (note that the QR decomposition of $(\Id_T \otimes \Hm)$ requires only the decomposition of the $N\times M$ channel matrix $\Hm$). 

\section{The Proposed Codes}\label{ST-TAST}

\subsection{Methodology and Examples} 
The new codes use the threading in \eqref{thread-layering} for $T = 2M+L-1$ where $L$ is a non negative integer that allows more tradeoff between complexity and rate as explained in the sequel. We encode $K=M(M + L)$ QAM information symbols, $\uv \defines (u_1, \ldots,u_K)^{\transp}$ into an $M\times T$ code word matrix $\Tm(\uv)$. 
In order to have an upper triangular equivalent code matrix, each SISO encoder creates an ISI-like channel in each thread. 
For simplicity of presentation, the code construction is explained first for $L=0$ (i.e., $T=2M-1$, and $K=M^2$). 
Let $\Sm(\uv)$ be a full diversity, full rate  $M\times M$ optimized TAST code (e.g., \cite{elgamal:damen02,simpletast06,Rajan1,perfectstc}). 
Then the new TAST code is constructed as follows.
 The first $M\times M$ block of the new code is given by
\begin{equation}\label{first}
\left(\Sm_1(u_1,u_2, \ldots, u_M,0,\ldots,0),\Sm_2(u_1,u_2,\ldots,u_{2M},0,\ldots,0),\ldots,\Sm_M(u_1,u_2,\ldots,u_{M^2})\right)
\end{equation}
where $\Sm_j$ is the $j$-th column of matrix $\Sm$. That is, at each symbol period, $M$ new QAM symbols enter the encoder until a total of $M^2$ symbols are encoded.  
The second and last $M\times (M-1)$ block is given by
\begin{equation}\label{last}
\left(\Sm_1(0,\ldots, 0, u_{M+1}, \ldots, u_{M^2}),\Sm_2(0,\ldots, 0, u_{2M+1},\ldots,u_{M^2}),\ldots,\Sm_{M-1}(0,\ldots,0,u_{M^2-M+1},\ldots,u_{M^2})\right).
\end{equation}
That is, at each symbol period, $M$ symbols that appeared $M$ times each in the code, exit the encoder. 

For example, for $M=2$ and $T=3$, using the Golden code \eqref{goldencode} as the constituent code, one has
\begin{equation}
\Tm(\uv) \defines \frac{1}{\sqrt{5}}\left[\begin{array}{ccc}
\alpha(u_1) & \phi^{1/2} \sigma(\alpha(u_2+\theta u_3)) & \alpha(\theta u_4) \\
\phi^{1/2} \alpha(u_2) & \sigma(\alpha(u_1 + \theta u_4)) & \phi^{1/2} \alpha(\theta u_3)
\end{array}\right],
\end{equation}
with $\alpha$, $\theta$, $\phi$ and $\sigma$ as in \eqref{goldencode}.
The rate of this code is $R = \frac{K}{T} = \frac{4}{3}$ complex symbols per channel use and it is easy to show that it achieves full diversity. The equivalent code matrix $\Mcb$ is upper triangular as can be seen from transforming $\Tm(\uv)$ into a vector by stacking its columns one after the other. 
\begin{equation}
\mbox{vec}(\Tm(\uv)) = \underbrace{\left( \begin{array}{cccc}
\alpha & 0 & 0 & 0 \\
0& \phi^{1/2} \alpha & 0 & 0 \\
0& \phi^{1/2} \sigma(\alpha) & \sigma(\alpha \theta) & 0 \\
\sigma(\alpha) & 0 & 0& \sigma(\alpha \theta) \\
0& 0 & 0& \alpha(\theta) \\
0&0&\phi^{1/2}\alpha\theta & 0 \end{array}\right)}_{\Mcb} \left(\begin{array}{c}
u_1\\
u_2\\
u_3\\
u_4\end{array}\right)
\end{equation}

To increase the rate while keeping the upper-triangular property of the equivalent code matrix $\Mcb$ and the full diversity of the code, one can add $L$ columns between the first and last block of the new TAST code, where for each added column, $M$ symbols exit the $M$ SISO encoders and $M$ new QAM symbols enter such that each of the $K$ QAM symbols {\bf appears exactly in} $M$ places {\bf in only one} of the $M$ thread. For example, for the above code with $L=1$, one obtains

\begin{equation}
\Tm(\uv) \defines \frac{1}{\sqrt{5}}\left[\begin{array}{ccccc}
\alpha(u_1) & \phi^{1/2} \sigma(\alpha(u_2+\theta u_3)) & \alpha(u_4 + \theta u_5) & \phi^{1/2} \sigma(\alpha \theta u_6)  \\
\phi^{1/2} \alpha(u_2) & \sigma(\alpha(u_1 + \theta u_4)) & \phi^{1/2} \alpha(u_3 + \theta u_6) &  \sigma(\alpha \theta u_5 )
\end{array}\right],
\end{equation}
which has a rate of $R=\frac{6}{4}$. 

The general construction is summarized as follows:
\begin{enumerate}
\item A $K\times 1$ QAM information symbol vector $\uv$ is encoded into an $M\times T$ space-time code word matrix $\Tm(\uv)$, where $K=M(M+L)$ and $T=2M+L-1$ with $L$ a non negative integer. 
\item The new $M\times T$ TAST code $\Tm$ is based on concatenating the columns of a constituent $M\times M$ full diversity and full rate TAST code $\Sm$ such that in the code word matrix $\Tm$, the columns of $\Sm$ are repeated cyclically, i.e., 
$\Tm_{j} = \Sm_{\lfloor j-1 \rfloor_M + 1}, \, \, j=1,\ldots, T$. 
\item After the first $M$ columns given in \eqref{first}, one can add $L$ columns to increase the rate and then one requires $M-1$ columns for termination (in order to guarantee the same diversity level for all the symbols). The addition of extra columns respects the following rule: {\bf each QAM symbol should appear exactly in $M$ positions in only one thread.} Mathematically, the $\ell$-th added column with $1 \leq \ell \leq L$ is of the form
\begin{equation}
\Tm_{M+\ell} = \Sm_{\lfloor M+\ell-1 \rfloor_M + 1}(u_{\ell M+1},\ldots,u_{\ell M+M^2})
\end{equation}
where at each columns, $M$ symbols that already appeared $M$ times each exit the encoder and $M$ new symbols enter.
\item The last $M-1$ columns are for termination where no new symbols enter the code and $M-1$ blocks of $M$ symbols exit the encoder ($M$ symbols at each column): 
\begin{eqnarray}\label{last-p}
\Tm_{T-M+2} &=& \Sm_{\lfloor T-M+1 \rfloor_M + 1} (0,\ldots, 0, u_{LM+M+1}, \ldots, u_{M^2+LM}) \\
\Tm_{T-M+3} &=& \Sm_{\lfloor T-M+2\rfloor_M+1} (0,\ldots, 0, u_{LM+ 2M+1},\ldots,u_{M^2+LM}) \\
\vdots \\
\Tm_{T} &=& \Sm_{\lfloor T-1\rfloor_M+1}(0,\ldots,0,u_{LM + (M-1)M + 1},\ldots,u_{M^2 + LM}).
\end{eqnarray}
\end{enumerate}

Note that the space-time code $\Sm(\vv)$ takes input vector of length $M^2$, and therefore, the number of zero elements in $\vv$ using the above notations equal $M^2$ minus the number of nonzero elements. 

One has the following.

\begin{theorem}
The new TAST code constructed as above over $M$ transmit antennas and $T=2M+L-1$ transmission intervals has a rate of $R=\frac{M(M+L)}{2M+L-1}$ symbols per channel use and achieves full diversity whenever the constituent TAST code achieves full diversity. 
\end{theorem}

\begin{proof}
To prove full diversity, we consider the rank of the difference between distinct code word matrices, $\Tm(\uv_1) - \Tm(\uv_2)$ with $\uv_1\neq \uv_2$. Equivalently, we can consider $\Tm(\uv)$ with $\uv \neq \zerov$. Since $\uv \neq \zerov$, it follows that at least one $u_j$ is nonzero. By construction, each $u_j$ appears exactly in $M$ positions along one thread in $\Tm$. Now, consider the square $M\times M$ submatrix that contains those positions of $u_j$. By construction, it is the constituent TAST matrix $\Sm$ (up to a column permutation). Since the constituent TAST code achieves full diversity, it follows that the rank of the considered submatrix equals $M$ and the new TAST code $\Tm$ achieves full diversity. 
\end{proof}

The following example illustrates the idea of creating an ISI-channel in each thread in order to achieve diversity while reducing detection complexity. 

\begin{example}
Let $M=3, L=1$ giving $T=2M-1+L=6$ and $K=M(M+L)=12$. Let $u_1, \ldots, u_K \in $ some QAM constellation in $\ZZ[i]$, the ring of Gaussian integers.  Further, let $\theta$  and $\phi$ be the corresponding Diophantine numbers of the TAST construction \cite{simpletast06} (i.e., $\{1,\theta,\theta^2\}$ are algebraically independent over $\QQ(i)$ and $\{1,\phi,\phi^2\}$ are algebraically independent over $\QQ(i,\theta)$). Then, the following construction is one example of a simplified tree-structured TAST code:
\begin{equation}
\left( \begin{array}{cccccc}
u_1 & \phi^{2/3} (u_4+\theta u_3) & \phi^{1/3}(u_7 + \theta u_6 + \theta^2 u_2) & u_{10}+\theta u_9 + \theta^2 u_{5} & \phi^{2/3} (\theta u_{12}+\theta^2 u_8) & \phi^{1/3} \theta^2 u_{11} \\
\phi^{1/3} u_2 & u_5 + \theta u_1 & \phi^{2/3}(u_8 + \theta u_4 + \theta^2 u_3) & \phi^{1/3} (u_{11}+\theta u_7 + \theta^2 u_{6}) & \theta u_{10} + \theta^2 u_{9} & \phi^{2/3} \theta^2 u_{12} \\
\phi^{2/3} u_3 & \phi^{1/3} (u_6 + \theta u_2) & u_9 + \theta u_5 + \theta^2 u_1 &\phi^{2/3}(u_{12} + \theta u_8 + \theta^2 u_{4} & \phi^{1/3} (\theta u_{11}+\theta^2 u_7) &   \theta^2 u_{10} \end{array} \right).
\end{equation}
Note that each symbol $u_j$, $j=1,\ldots,12$ appears exactly $M=3$ times in one of the $3$ threads, guaranteeing thus its full diversity level. Under this representation, the simplified tree-structured TAST code can be considered as a finite-state machine (i.e., convolutional encoder over the complex numbers field) where the first thread is generated by the convolutional encoder $(1,\theta,\theta^2)$, the second thread by $\phi^{1/3}(1,\theta,\theta^2)$ and the third thread by $\phi^{2/3}(1,\theta,\theta^2)$. Note also that, similar to the fractional rate loss in convolutional codes due to equal error protection, we have a fractional rate loss that vanishes when the block length increases. 

\end{example}

\subsection{New Constructions Properties}
The main property of the new codes is that the equivalent code matrix $\Mcb$ is upper triangular which is easily shown form the construction method. This largely simplifies the QR decomposition of the composite channel and code matrix $(\Id_T \otimes \Hm ) \Mcb$ which is essential to ML detection using sphere or sequential decoding or other sub-optimal near-ML decoding algorithms such as the Babai decoder \cite{dec03}. 

This reduction in complexity is achieved at the expense of rate reduction. While the constituent TAST code has a full rate of $M$ symbols per channel use, the new construction achieves a rate of $R=\frac{M(M+L)}{2M+L-1}$, which only asymptotically tends to the full rate when $L$ increases. 
 
 As for the other code properties, such as peak-to-average power ratio or the coding gain, the new construction inherits the properties of the constituent TAST code. In particular, one can show that if the constituent TAST code $\Sm$ has a non-vanishing determinant (NVD)\footnote{The NVD property is useful for achieving the optimal diversity-multiplexing tradeoff, or DMT,  in MIMO systems (see for example \cite{perfectstc} and references therein).} then the new construction $\Tm$ also has the NVD property. This can be seen from the construction of $\Tm$ where it has to contain (up to column permutations) the submatrix $\Sm$ with at least one nonzero thread, which gives the NVD property.  While the new codes are not DMT optimal for finite $L$ due to their rate loss, using the NVD property, one can show their optimality in the limit of $L$. 

The new TAST codes with simplified tree structures inherit other desirable properties form their TAST brethren such as their availability for any number of transmit and receive antennas, any channel-state information at the receiver and transmitter and any input alphabets \cite{elgamal:damen02}. 

Finally, note that one can increase the rate (while keeping the upper triangular form of $\Mcb$) by cutting the termination tail of $M-1$ columns in $\Tm$, in which case one obtains unequal error protection (that results from unequal diversity levels) for the different QAM symbols. 

\section{Simulation Results}\label{conc}

\begin{figure}
  \includegraphics[width=1.0 \textwidth]{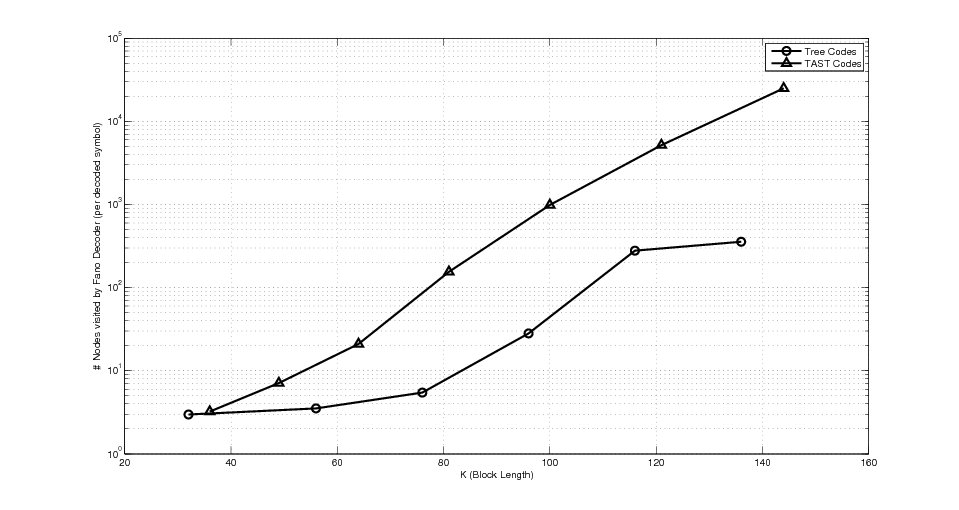}
  \caption{The effect of block length on the complexity of the Fano Deocoder applied to the proposed tree codes and the original TAST codes.}
\label{Nodes_VS_K}
\end{figure}

\begin{figure}
  \includegraphics[width=1.0 \textwidth]{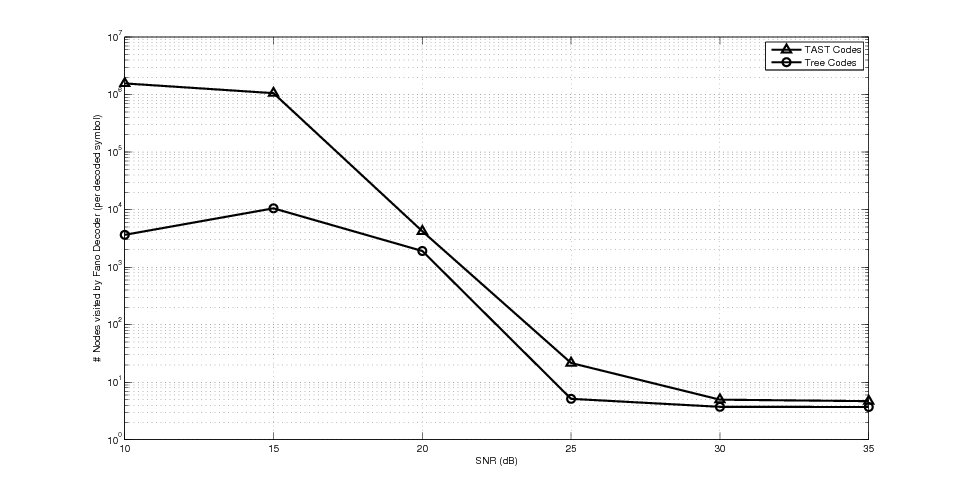}
  \caption{The complexity of the Fano Deocoder applied to the proposed tree codes and the original TAST codes.}
\label{Nodes_VS_SNR}
\end{figure}

\begin{figure}
  \includegraphics[width=1.0 \textwidth]{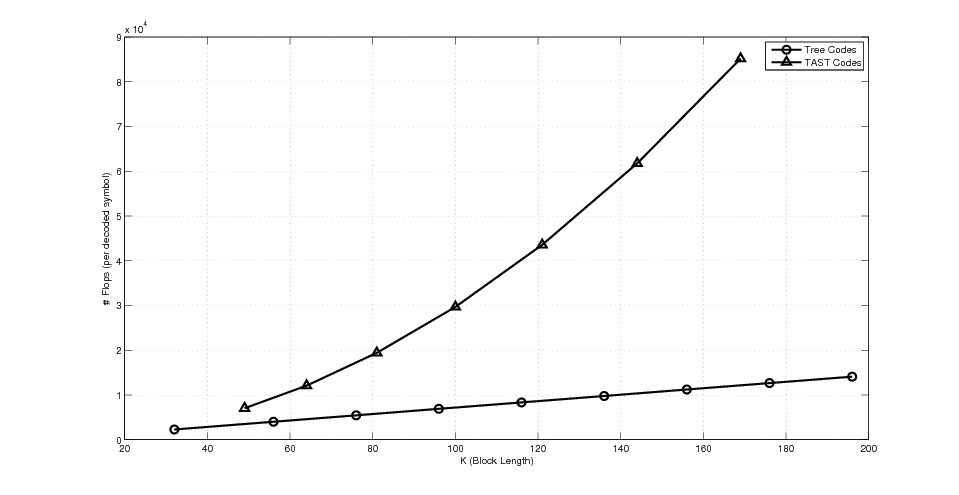}
  \caption{The complexity of the QR Decomposition using Givens Rotations of the proposed tree codes and the original TAST codes.}
\label{QR_VS_K}
\end{figure}

Except as noted to the contrary, we employ BPSK modulation in all our simulations. Simulation results presents the degradation in complexity achieved through the proposed codes over the original TAST codes in both stages, the preprocessing stage (represented in the QR decomposition) and the tree search stage (represented in the number of nodes visited by the Fano Decoder). Unless specified otherwise, the Diophantine numbers are chosen according to \cite{simpletast06} to be algebraic integers such that the TAST code achieves full diversity.
Fig~\ref{Nodes_VS_K} and ~\ref{Nodes_VS_SNR} show the reduction in complexity in the search stage achieved by the proposed tree codes over the original TAST codes which is clear especially for large block lengths and also at low SNR. This is because of the sparsity of the composite $\mathcal{HG}$ matrix which facilitates the role of the Fano Decoder.\\
Fig ~\ref{QR_VS_K} shows the main reduction in complexity in the preprocessing stage (the QR decomposition). This was mainly achieved by applying QR decomposition through Givens rotations which introduces one zero at a time which helps in investing the ready zeros in the tall upper-triangular composite matrix $\mathcal{HG}$. The exact number of flops needed per decoded symbol is calculated to be $[N(M - 1) + \frac{N-M}{2M}]K + [\frac{N-M}{2M}]K^2$. So, it's obvious that at $M=N$, the number of flops required per decoded symbol reduces to $M(M-1)K$ which is of $\mathcal{O}(K)$ instead of $\mathcal{O}(K^2)$ for the original TAST codes. The figure shows this linear relation with the block length.

\section{Conclusion}\label{conc}
In this paper, we have presented a new methodology for constructing TAST codes with simplified tree structure. The new codes have the interesting property of upper triangular equivalent code matrix which simplifies the subsequent tree search algorithms such as sequential decoders. The complexity reduction comes at a price of a fractional rate loss, when compared to the full rate TAST codes, that vanishes when the block size increases. The proposed codes present an interesting compromise between rate and complexity which situates them as an attractive alternative of low rate low complexity orthogonal and quasi-orthogonal designs and high-rate high complexity TAST codes or codes from cyclic division algebra. Further simplification of the proposed codes is under consideration as well as their extension to distributed relay network  \cite{delay-tolerant}.

\end{document}